\documentclass[runningheads]{llncs}

\pdfoutput=1

\usepackage[T1]{fontenc}
\usepackage{graphicx}
\usepackage{amsmath}
\usepackage{amssymb} 
\usepackage{mathpartir}
\usepackage{simplebnf}
\usepackage{breqn}
\usepackage{bm}
\usepackage{orcidlink}
\usepackage[misc]{ifsym}
\usepackage{xcolor}
\usepackage[frozencache,cachedir=minted-cache]{minted}
\usepackage[figurename=Figure]{caption}
\usepackage{subcaption}
\usepackage{enumitem}

\DeclareMathSymbol{\qm}{\mathalpha}{operators}{"3F}
\DeclareMathAlphabet{\mathbbold}{U}{bbold}{m}{n}

\AfterEndEnvironment{minted}{\par\vspace{-\baselineskip}}

\newcommand{\bstep}[3]{#1 \vdash #2 \Rightarrow #3}
\newcommand{\bstepb}[3]{#1 \vdash #2 \hookrightarrow #3}
\newcommand{\ite}[3]{\textbf{if}\; #1\; \textbf{then}\; #2\; \textbf{else}\; #3}
\newcommand{\true}{\textbf{True}}
\newcommand{\false}{\textbf{False}}
\newcommand{\none}{\textbf{None}}
\newcommand{\some}[1]{\textbf{Some}\;#1}
\newcommand{\either}[2]{\textbf{either}\;#1\;\textbf{or}\;#2}
\newcommand{\lam}[3]{\boldsymbol{\lambda}^{#2}_{#1} \bm{.}\; #3}
\newcommand{\fby}[2]{#1\; \textbf{fby}\; #2}
\newcommand{\pre}[1]{\textbf{pre}\; #1}
\newcommand{\arrow}[2]{#1\, \bm{\rightarrow}\, #2}
\newcommand{\tup}[2]{(#1\;\textbf{,}\;\ldots\;\textbf{,}\;#2)}

\newcommand{\fqm}{\framebox{$\mathbbold{\qm}$}}

\begin{document}

\title{Mimosa: A Language for Asynchronous Implementation of Embedded Systems Software}
\titlerunning{The Mimosa Language}

\author{Nikolaus Huber\inst{1}${}^{(\text{\Letter})}$\orcidlink{0000-0002-1616-0602} \and
Susanne Graf\inst{1,2}\orcidlink{0000-0003-4354-6807} \and
Philipp Rümmer\inst{1,3}\orcidlink{0000-0002-2733-7098} \and
Wang Yi\inst{1}\orcidlink{0000-0002-2994-6110}}

\authorrunning{N. Huber et al.}
\institute{Uppsala University, Uppsala, Sweden\\
\and
Université Grenoble Alpes, CNRS, Grenoble INP, VERIMAG, France\\
\and
University of Regensburg, Regensburg, Germany}

\maketitle

\begin{abstract}

This paper introduces the Mimosa language, a programming language for the design and implementation of asynchronous reactive systems, describing them as a collection of time-triggered processes which communicate through FIFO buffers. Syntactically, Mimosa builds upon the Lustre data-flow language, augmenting it with a new semantics to allow for the expression of side-effectful computations, and extending it with an asynchronous coordination layer which orchestrates the communication between processes. A formal semantics is given to both the process and coordination layer through a textual and graphical rewriting calculus, respectively, and a prototype interpreter for simulation is provided.

\end{abstract}
\section{Introduction}
\label{sec:introduction}

The synchronous paradigm is an established method for designing, implementing, and verifying the software layer of embedded systems. Languages like Lustre~\cite{halbwachs-lustre} (and its commercial implementation SCADE~\cite{colaco-scade}), Signal~\cite{benveniste-signal}, and Esterel~\cite{berry-esterel} have been successfully used for the implementation of various safety-critical embedded control applications. However, on complex execution platforms, such as multi- and many-core processors or distributed architectures, the stringent timing constraints required by the execution models of such languages becomes increasingly difficult to uphold. Problems further arise when trying to update synchronous systems, as newly added software components might put additional constraints on the global system tick.

The MIMOS computational model~\cite{wang-mimos} offers a new way of expressing asynchronous software designs for embedded systems. It describes a program as a network of periodically triggered processes, which communicate through FIFO buffers. The MIMOS project~\cite{wang-mimos-tool} is an ongoing effort to develop a set of tools to allow designing, implementing, and verifying embedded software on top of the MIMOS model.

This paper deals with the first layer of this framework, the user-facing programming language, the MIMOS Application (Mimosa) language. Mimosa consists of a small kernel, which covers the implementation of the (input/output) function of each process, as well as the coordination layer. Our main contribution is the definition of the formal semantics of both layers, utilizing different formal frameworks for each.

The paper is structured in the following way. Section~\ref{sec:mimos} gives a short overview of the MIMOS computational model, and defines the subset which is currently covered by Mimosa. Section~\ref{sec:examples} gives an introduction to Mimosa by showcasing various examples, before defining the abstract syntax and formal semantics of the language in Section~\ref{sec:language}. We present a simulator for Mimosa in Section~\ref{sec:simulation}. Finally, we report on related work in Section~\ref{sec:related-work} and conclude with a list of ongoing and future efforts in Section~\ref{sec:conclusion}.
\section{The MIMOS model}
\label{sec:mimos}

The MIMOS model~\cite{wang-mimos} builds upon the well-known framework of Kahn process networks (KPNs)~\cite{kahn-networks}. In a KPN, different software components communicate exclusively through unbounded FIFO buffers, where reading from a buffer is blocking, i.e., if a component tries to read from an empty buffer it is suspended until data is available. Kahn showed, that the output of such a system (i.e., the history of values appearing in each buffer) is deterministic, independently of the scheduling order of the individual components.

In real-time systems, a bounded delay between input and output must be guaranteed, which due to the independence of any scheduling order is difficult to prove for a given KPN. MIMOS addresses this issue by assigning to each component a release pattern (i.e., an infinite series of increasing time-tags), which marks the points in time at which each respective component will try to execute. In case a component does not have all its required input at a release time point, it remains idle until its next release.

Through this assignment of release patterns, different extensions to KPNs are possible, without affecting (timed) determinism. One such extension presented in \cite{wang-mimos} is registers, which instead of buffering allows expressing a latest-value semantics.

Another possible extension is optional inputs~\cite{wang-mimos-tool}. In a KPN, a process trying to read from a buffer must always wait until data is available. With timed release patterns, it is possible to define an input as optional, where the value is subject to availability in the buffer.

\begin{figure}[t]
    \centering
    \includegraphics[width=0.75\linewidth]{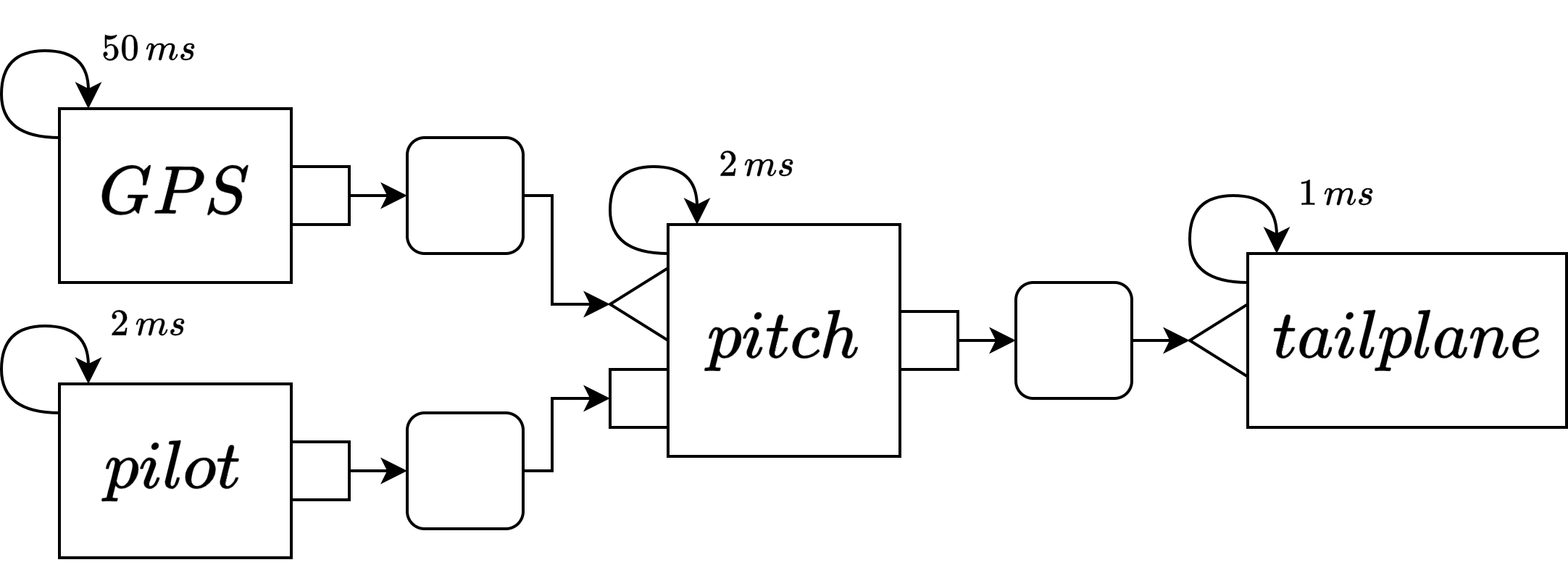}
    \caption{Airplane pitch control system, adapted from~\protect\cite{henzinger-giotto}}
    \label{fig:pitch-example}
\end{figure}

Figure~\ref{fig:pitch-example} shows a simplified example of a MIMOS network for the pitch control of an airplane. It consists of 4 nodes, one handling the input from the GPS subsystem, one polling the stick-control of the pilot, the actual pitch control algorithm, and the driver for the tailplane. As the GPS system is slow, it runs with a lower frequency (i.e., higher period) than the pitch control algorithm. The tailplane driver runs fastest, as it must react to minute changes in the mechanical parts of the aircraft. By utilizing optional inputs at the pitch control node for the GPS signal, as well as at the tailplane driver for the control signal, each node can run at its required frequency without constraining the rest of the system.

In the original MIMOS paper~\cite{wang-mimos}, a fix-point semantics has been presented, however, it neither refers to a particular programming language for the implementation of the processes, nor does it provide a textual syntax for the definition of the coordination layer. Our contribution in this paper is the formal definition of Mimosa, which equips MIMOS with a prototype programming language kernel.

As the main focus of this work is on the formal semantics of the language, we have put certain restrictions on the particular MIMOS model that we cover. For simplicity, registers are currently not modelled by our semantics, we discuss in Section~\ref{sec:conclusion} on how they can be added. We also restrict the release patterns of the processes to periodic releases, with an implicit deadline at the end of each period (i.e., if a component with period $p$ reads its required inputs at time $t$, the outputs are written at $t + p$).
\section{From Lustre to Mimosa}
\label{sec:examples}

\begin{figure}[p]
\centering
\includegraphics[width=0.8\linewidth]{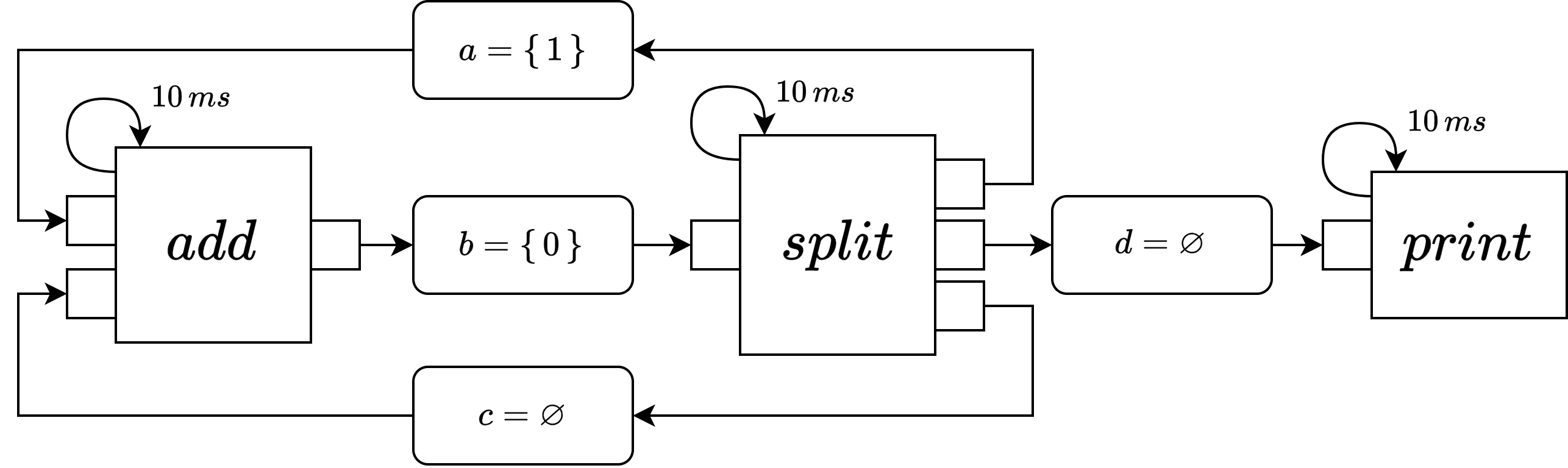}
\captionof{figure}{Fibonacci example (graphical)}
\label{fig:fibonacci}
\begin{minted}[frame=lines, linenos]{text}
step print_int (_ : int) --> ()
step add (x, y) --> z { z = x + y }
step split inp --> (o1, o2, o3) { o1, o2, o3 = inp, inp, inp }

channel a : int = { 1 }
channel b : int = { 0 }
channel c : int
channel d : int

node add implements add (a, c) --> (b) every 10ms
node split implements split (b) --> (a, d, c) every 10ms
node print implements print_int (d) --> () every 10ms
\end{minted}
\captionof{listing}{Fibonacci example (Mimosa)}
\label{code:fib}
\end{figure}

\begin{figure}[p]
\centering
\includegraphics[width=0.5\linewidth]{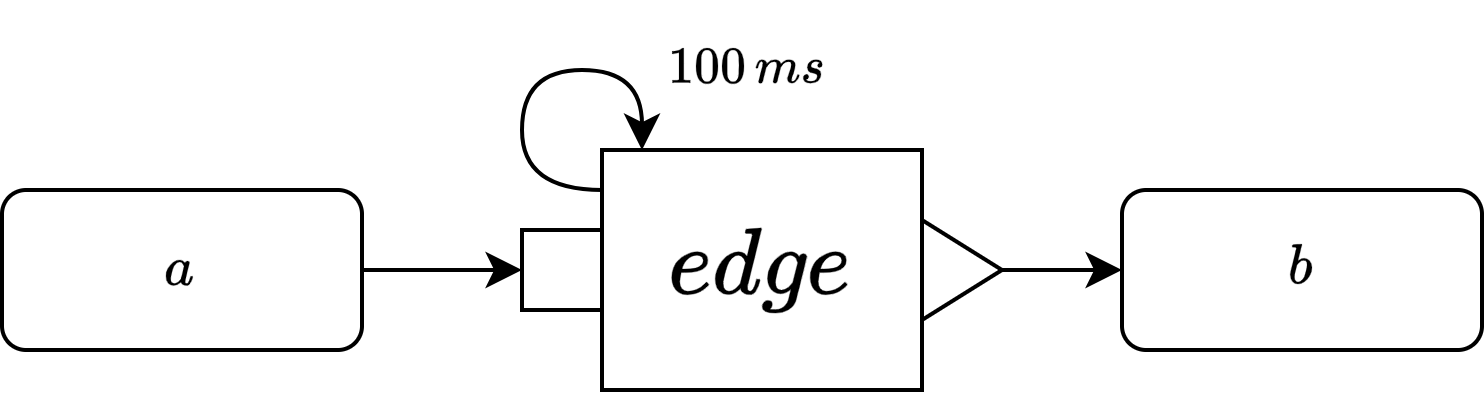}
\captionof{figure}{Edge detector (graphical)}
\label{fig:enter-label}
\begin{minted}[frame=lines, linenos]{text}
step edge_detect (in : bool) --> (out : bool?)
{
    pre_in = in -> pre in;
    out = if !pre_in && in then (Some true)
          else if pre_in && !in then (Some false)
          else None;
}

channel a : bool
channel b : bool

node edge implements edge_detect (a) --> (b?) every 100ms    
\end{minted}
\captionof{listing}{Edge detector (Mimosa)}
\label{code:edge-detector}
\end{figure}

In Lustre, any variable or expression denotes a (conceptually infinite) stream of values, so that a variable $x$ actually represents a sequence ($x_0,\, x_1,\, \ldots)$. A Lustre program transfers a set of input streams into a set of output streams. Streams are defined through equations, each of the form $x = e$, where the expression $e$ is formed from constants (representing infinite sequences of the same value), variables (i.e., references to other streams), and stream operators (like the usual unary and binary arithmetic and logic functions extended to operate point-wise on sequences) applied to argument streams.

In addition, Lustre defines a set of \emph{sequence operators} (\textbf{pre}, \textbf{fby}, and $\bm{\rightarrow}$). If $x = (x_0,\, x_1,\, \ldots)$, then $\pre{x}$ refers to the value of $x$ from the previous cycle, i.e., $\pre{x} = (\perp,\, x_0,\, x_1,\, \ldots)$. The value at the first cycle is undefined ($\perp$). If $y = (y_0,\, y_1,\, \ldots)$, then $\arrow{x}{y} = (x_0,\, y_1,\, y_2,\, \ldots)$ and $\fby{x}{y} = (x_0,\, y_0,\, y_1,\, \ldots)$. The initialization operator $\bm{\rightarrow}$ is often used to initialize the first element of a stream under the \textbf{pre} operator (e.g., $\arrow{0}{\pre x}$). 

The condition operator is also extended to work point-wise on streams. Conceptually, $\ite{e_c}{e_t}{e_e}$ always evaluates both $e_t$ and $e_e$, and then selects according to the value of $e_c$. This allows Lustre to have full referential transparency, i.e., a variable inside an expression can always be substituted by its definition. In Lustre, expression evaluation is always assumed to be side effect free. Selective evaluation in Lustre is possible through the definition of multiple clocks, where the evaluation of an expression depends on the value of a boolean stream. However, experience has shown~\cite{halbwachs-lustre-at-work}, that working with multiple clocks is often perceived as too difficult to comprehend by the programmer, and therefore not often used.

Mimosa can be considered an extension to Lustre. Since each component can have its own period, we shift the interpretation of streams towards timed sequences of values in FIFO buffers. We call the components of the network \emph{nodes}, and each node transfers (timed) input sequences into (timed) output sequences. Analogous to Lustre, sequences are defined through equations. However, we do not assume referential transparency, and therefore allow expression evaluation to cause side effects. The only effect we currently model is state (i.e., memory), however, as we explain in Section~\ref{sec:conclusion} our semantics lays the foundation for dealing with other effects as well. In the remainder of this section, we introduce Mimosa by showing examples of concrete programs.

A Mimosa program is a collection of top-level definitions consisting of \emph{steps}, \emph{channels}, and \emph{nodes}. A step is an elementary unit of computation, it does not have any perception of time. Steps are equivalent to function definitions in other programming languages. A step can be used inside another step, but it must be possible, during compilation, to order them in a way so that no two steps are mutually recursive. 
A node is an instantiation of a step as a periodically triggered process. It has a defined period, and port definitions to declare how it communicates with other nodes.
A channel is a FIFO buffer, through which nodes communicate with each other. Each channel is connected to exactly one writing and one reading node.
A port of a node may be marked as \emph{optional}. If an input port to a node is marked as optional, the node can be executed even if the input channel is empty. Similarly, if an output port is marked as optional, it may happen that after a node has executed its step, there is no output written to the connected channel.
 
Figure~\ref{fig:fibonacci} shows an example network adapted from~\cite{geilen-kahn-reactive} which calculates the Fibonacci sequence. It consists of three nodes, \texttt{add}, \texttt{split}, and \texttt{print}, each trying to execute every $10\, ms$. 
Listing~\ref{code:fib} shows the equivalent Mimosa program. It starts with the definition of the steps, where the first is a \emph{step prototype}. It defines only the name of the step and its signature, which will later be provided externally. For normal steps, the step signature is followed by a set of equations (same as in Lustre). The Fibonacci example defines the steps \texttt{add} and \texttt{split}, which have a trivial definition (\texttt{split} could have been equivalently defined through three separate equations). Similar to Lustre, the order of equations is not significant, as the compiler orders them automatically, and rejects a program with cyclic dependencies between equations.

After the step definitions, the channels are defined, of which two have initial elements. Finally, the three nodes are defined as instantiations of the before mentioned steps. They list which step they implement, their period, and their connections (which refer to the channels defined before).

The code shown in Listing~\ref{code:edge-detector} showcases the use of an optional output. The node \texttt{edge} outputs \texttt{true} whenever a rising edge in the boolean input sequence is detected, \texttt{false} when a falling edge is detected, and nothing otherwise. This can be done by wrapping the output in an option type (\texttt{out : bool?}) and declaring the output port as optional (\texttt{b?}), which means that even though the output of the step \texttt{edge} is \texttt{bool?}, the type of the channel is still \texttt{bool}. This example also illustrates the use of the \textbf{pre} and $\bm{\rightarrow}$ operators to define a memory cell (in this case, to remember the value of the input from the previous execution).

Figure~\ref{fig:pin-example} shows an example of how this edge detector with optional output may be used in a system. Assuming \texttt{pin} polls the current level of a pin, the \texttt{edge} node then detects rising and falling edges, which it communicates to a \texttt{controller}, which can run at a lower frequency (assuming that the pin-level switches infrequently). While such a system can be implemented in the synchronous paradigm as well, the base-clock (i.e., the shortest time tick) would be constrained by the component with the smallest period. Due to the asynchronous character of Mimosa, this issue does not arise, as individual components can have different periods without constraining the overall system.

\begin{figure}[t]
    \centering
    \includegraphics[width=0.75\linewidth]{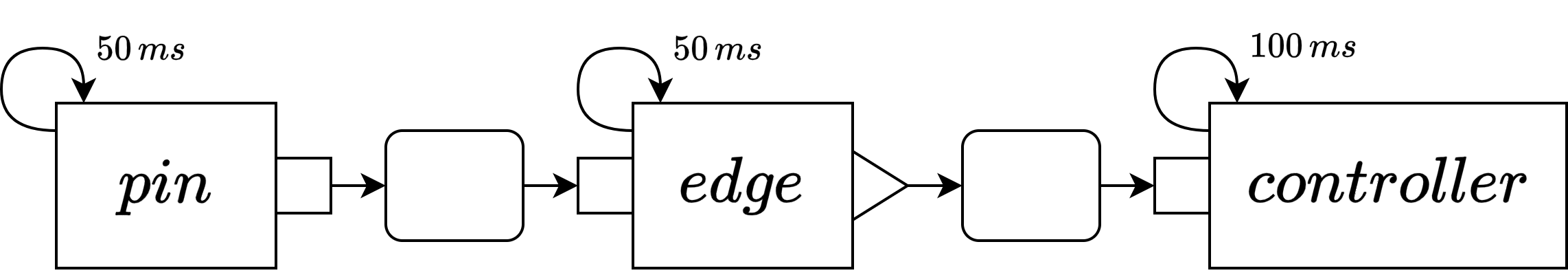}
    \caption{Usage example for edge detector}
    \label{fig:pin-example}
\end{figure}
\section{Syntax and semantics}
\label{sec:language}

We now formally define the semantics of Mimosa. As the language naturally divides into two separate layers (the step and coordination layer), we give semantics to each of them individually. We start by defining the abstract syntax and semantics of steps.

\subsection{Step layer}

\begin{figure}[th]
\centering
\fbox{
\begin{bnf}(comment = {::},)[]
    $e$ : \textsf{Expression} ::= 
    | $x$ :: Variable 
    | \texttt{c} :: Constant 
    | $e\bm{,} \ldots \bm{,} e$ :: Tuple 
    | $\pre{e}$ :: Pre 
    | $\fby{e}{e}$ :: Followed-by 
    | $\arrow{e}{e}$ :: Initialized-by 
    | $e$ $e$ :: Application 
    | $\ite{e}{e}{e}$ :: Conditional
    | $\none$ :: None
    | $\some{e}$ :: Some
    | $\either{e}{e}$ :: Option match
    | $\lam{p}{p}{b}$ :: Abstraction 
    ;;
    $p$ : \textsf{Pattern} ::= 
    | $x$ :: Variable pattern
    | $p\bm{,} \ldots\bm{,} p$ :: Tuple pattern
    ;;
    $b$ : \textsf{Body} ::= 
    | $[\,p_i \bm{=} e_i\,]^n$ :: List of $n$ equations 
    ;;
\end{bnf}}
\caption{Abstract syntax of Mimosa expressions}
\label{fig:abstract-syntax-steps}
\end{figure}

The abstract syntax of Mimosa expressions shown in Figure~\ref{fig:abstract-syntax-steps} is actually slightly more expressive than the one defined by the concrete syntax of Mimosa, where we only allow the definition of function abstractions (i.e., steps) on the top-level. The abstract syntax also allows nested (anonymous) functions, which allows us to treat functions as regular expressions. We expand on this choice in Section~\ref{sec:conclusion}.

A Mimosa expression is either a variable, a constant, a tuple construction, a definition of memory (using the \emph{memory operators} \textbf{pre}, \textbf{fby}, and $\bm{\rightarrow}$), an application of a function to an argument, a conditional, the construction or destruction of an optional expression, or a function abstraction.

We assume a suitable set of constants and a corresponding set of basic steps for arithmetic and boolean operations to be defined in a standard library. Mimosa uses a Hindley-Milner-style type system~\cite{milner-type-polymorphism,damas-type-schemes}, which we omit for brevity.

Patterns either refer to a single variable, or to a tuple of sub-patterns. In practice, we also allow  \textunderscore{} to mean a pattern matching any value.

The abstraction $\lam{p_{in}}{p_{out}}{[p_i = e_i]^n}$ defines a function that requires a value compatible with pattern $p_{in}$, and returns a value according to pattern $p_{out}$. It is defined through a set of $n$ equations, where each one is given by a left-hand pattern $p_i$ and a right-hand expression $e_i$. It is assumed that equations are ordered according to dependency, are causal (i.e., any cyclic dependency is broken by a \textbf{pre}), and properly initialized (see Appendix~\ref{app:init-analysis}). 

An expression can be evaluated under an environment $\Gamma$ which maps names to values. We define two operations, \emph{projection} and \emph{update}, on environments:

\paragraph{Projection}
The operation $\Gamma \Downarrow_{p}$ returns the bindings of the variable names given by pattern~$p$. For example:
\begin{align*}
(x \mapsto 1, y \mapsto 2) \Downarrow_x\; &=\; 1 \\ 
(x \mapsto 1, y \mapsto 2) \Downarrow_{(x, y)}\; &=\; (1, 2) 
\end{align*}

\paragraph{Update}
The operation $\Gamma \Uparrow_p^x$ returns a new environment where the value x is mapped to the name(s) given by pattern $p$. For example:
\begin{align*} 
(x \mapsto 1, y \mapsto 2) \Uparrow_z^3\; &=\; (x \mapsto 1, y \mapsto 2, z \mapsto 3) \\ 
(x \mapsto 1, y \mapsto 2) \Uparrow_{(x, y)}^{(3, 4)}\; &=\; (x \mapsto 3, y \mapsto 4) \label{eq:destruct}
\end{align*}

As illustrated by the example above, updating an environment is destructive, i.e., previous bindings are lost.

The semantics of Lustre is usually defined denotationally as the unique least-fixed-point of the given set of stream equations. This works well for a language with full referential transparency. In Mimosa, however, the selective evaluation of certain sub-expressions makes it harder to find a denotational interpretation. We therefore express the semantics of expression evaluation in Mimosa operationally.

The evaluation relation $\Gamma \vdash e \Rightarrow v, e'$ expresses, that under an environment $\Gamma$, the expression $e$ evaluates to value $v$ (which is either a constant or a tuple of values), and an updated expression $e'$, which shall be evaluated the next time the current expression is run. The full set of evaluation rules is presented in Figure~\ref{fig:big-step}.

The rules for variable, constant, and tuple evaluation are trivial. The evaluation of $\pre{e}$ returns an undefined value $\perp$, where the next expression is the current value of $e$ initializing the updated $\pre{e'}$. This ensures that the current value of $e$ is always returned at the next evaluation cycle.

Both $\fby{e_1}{e_2}$ and $\arrow{e_1}{e_2}$ return the value of $e_1$, however, only in the second case is $e_2$ also evaluated, and its value discarded (this allows, analogously to Lustre, to remove the $\perp$ values of \textbf{pre} expressions).

The rules for conditional execution are again trivial. It is important to note, that unlike in Lustre, only one branch is evaluated, while the other is kept as is in the next expression.

The expression $\either{e_1}{e_2}$ evaluates the first sub-expression $e_1$, and in case it is of the form $\some{v}$ returns value $v$, otherwise it evaluates $e_2$. Similar to conditional execution, $e_2$ is only evaluated if needed.

To evaluate an application $e_1\; e_2$, the expression $e_1$ is mapped to a function abstraction $\lam{p_{in}}{p_{out}}{[p_i=e_i]^n}$, and the argument $e_2$ to a value $v$. The list of equations $[p_i=e_i]^n$ is then evaluated in an environment in which $v$ is bound according to $p_{in}$, which results in an updated list $[p_i=e_i']^n$, and a new environment $\Gamma'$. The value of the overall application is then the projection of $p_{out}$ onto $\Gamma'$, and the next expression is again an application, where the function term is updated with the new equation list, and the argument expression is $e_2'$.

\begin{figure}[p]
\fbox{
\begin{mathpar}
\inferrule 
    {\\}
    {\bstep{\Gamma}{x}{\Gamma \Downarrow_x, x}} 
    \; (\textsc{Var})

\inferrule 
    {\\} 
    {\bstep{\Gamma}{c}{c, c}} 
    \; (\textsc{Const}) 

\inferrule 
    {\bstep{\Gamma}{e_1}{v_1, e_1'} \\ \cdots \\ \bstep{\Gamma}{e_n}{v_n, e_n'}} 
    {\bstep{\Gamma}{\tup{e_1}{e_n}}{\tup{v_1}{v_n}, \tup{e_1'}{e_n'}}} 
    \; (\textsc{Tuple})  

\inferrule 
    {\bstep{\Gamma}{e}{v, e'}} 
    {\bstep{\Gamma}{\pre{e}}{\perp, (\arrow{v}{\pre{e'}})}}
    \; (\textsc{Pre})

\\

\inferrule 
    {\bstep{\Gamma}{e_1}{v_1, \_}} 
    {\bstep{\Gamma}{(\fby{e_1}{e_2})}{v_1,e_2}}
    \; (\textsc{Fby})

\inferrule 
    {\bstep{\Gamma}{e_1}{v_1, \_} \\ \bstep{\Gamma}{e_2}{\_, e_2'}}
    {\bstep{\Gamma}{(\arrow{e_1}{e_2})}{v_1, e_2'}}
    \; (\textsc{Init}) 

\inferrule 
    {\bstep{\Gamma}{e_1}{\true, e_1'} \\ \bstep{\Gamma}{e_2}{v, e_2'}}
    {\bstep{\Gamma}{(\ite{e_1}{e_2}{e_3})}{v, (\ite{e_1'}{e_2'}{e_3})}}
    \; (\textsc{IfTrue})

\inferrule 
    {\bstep{\Gamma}{e_1}{\false, e_1'} \\ \bstep{\Gamma}{e_3}{v, e_3'}}
    {\bstep{\Gamma}{(\ite{e_1}{e_2}{e_3})}{v, (\ite{e_1'}{e_2}{e_3'})}}
    \; (\textsc{IfFalse})

\inferrule 
    {\\}
    {\bstep{\Gamma}{\textbf{None}}{\textbf{None}, \textbf{None}}}
    \; (\textsc{None})

\inferrule 
    {\bstep{\Gamma}{e}{v, e'}}
    {\bstep{\Gamma}{\some{e}}{\some{v}, \some{e'}}}
    \; (\textsc{Some})

\inferrule 
    {\bstep{\Gamma}{e_1}{\textbf{Some}\; v, e_1'}}
    {\bstep{\Gamma}{(\either{e_1}{e_2})}{v, (\either{e_1'}{e_2})}}
    \; (\textsc{EitherSome})

\inferrule 
    {\bstep{\Gamma}{e_1}{\textbf{None}\;, e_1'} \\ \bstep{\Gamma}{e_2}{v, e_2'}}
    {\bstep{\Gamma}{(\either{e_1}{e_2})}{v, (\either{e_1'}{e_2'})}}
    \; (\textsc{EitherNone})

\inferrule
    {
        \bstep{\Gamma}{e_1}{\bm{\lambda}^{p_{out}}_{p_{in}} \bm{.}\; [p_i = e_i]^n, \_} \\
        \bstep{\Gamma}{e_2}{v_2, e_2'} \\ 
        \bstepb{\Gamma \Uparrow_{p_{in}}^{v_2}}{[p_i=e_i]^n}{[p_i=e_i']^n, \Gamma'} 
    }
    {\bstep{\Gamma}{e_1\; e_2}{\Gamma' \Downarrow_{p_{out}}, (\bm{\lambda}^{p_{out}}_{p_{in}} \bm{.}\; [p_i = e_i]^n)\; e_2'}}
    \; (\textsc{App})
    
\inferrule
    {\bstep{\Gamma'}{e_1}{v_1, e_1'} \\ \ldots \\ \bstep{\Gamma'}{e_n}{v_n, e_n'}}
    {\bstepb{\Gamma}{[p_i \bm{=} e_i]^n}{[p_i \bm{=} e_i']^n, \Gamma' = \Gamma \Uparrow_{p_1}^{v_1}\ldots\Uparrow_{p_n}^{v_n}}}
    \; (\textsc{Eqs})
\end{mathpar}}
\caption{Structural semantics for expression evaluation}
\label{fig:big-step}
\end{figure}
 
To illustrate the principle of equation evaluation, we can simplify the \textsc{Eqs} rule to the evaluation of a single equation. At first try, the rule may be written as follows:
\begin{mathpar}
    \inferrule
    {\bstep{\Gamma}{e}{v, e'}}
    {\bstepb{\Gamma}{(p = e)}{(p = e'), \Gamma \Uparrow_p^v}}
    \; (\textsc{Eqs'})
\end{mathpar}

To evaluate an equation, one has to evaluate the expression $e$ on the right-hand side, which results in a value $v$ and an updated expression $e'$. The result of the evaluation of the whole equation is then an updated equation $p = e'$, and a new environment in which the value $v$ is bound according to the pattern $p$ on the left-hand side of the given equation. However, this rule falls short when trying to evaluate an equation that refers to a previous value of itself. We can illustrate that on a simple example:
\begin{align*}
    x = \arrow{0}{\pre{x}}
\end{align*}

This defines a perfectly valid sequence of integers (i.e., the constant sequence of consecutive zeros), however, if we evaluate it according to the above given rule 
\begin{mathpar}
\inferrule* 
    { \inferrule* {{ \inferrule*{ }{ \bstep{\Gamma}{0}{0, 0} } \\ \inferrule*[right=($\star$)]{ \bstep{\Gamma}{x}{\fqm, x} }{ \bstep{\Gamma}{(\pre{x})}{\perp, (\arrow{\fqm}{\pre{x}})} }}}
      { \bstep{\Gamma}{(\arrow{0}{\pre{x}})}{0, (\arrow{\fqm}{\pre{x}})} }
    }
    { \bstepb{\Gamma}{(x = \arrow{0}{\pre{x}})}{(x =\arrow{\fqm}{\pre{x}}), \Gamma \Uparrow_x^0}}
\end{mathpar}
we see that there is an issue when trying to evaluate the \textbf{pre} at $(\star)$. In the example above, \small$\fqm\;$\normalsize represents a hole in the expression, for which we would need the value of $x$. But this value is not yet bound in $\Gamma$, as it is being defined by the equation currently under evaluation. Therefore, we need to evaluate the right-hand side expression of an equation under an environment that already has the final value of the evaluation bound according to the pattern on the left-hand side: 
\begin{mathpar}
    \inferrule
    {\bstep{\Gamma'}{e}{v, e'}}
    {\bstepb{\Gamma}{(p = e)}{(p = e'), \Gamma' = \Gamma \Uparrow_p^v}}
    \; (\textsc{Eqs''})
\end{mathpar}

\noindent With this rule, we get the correct result which we formulate in the lemma below:
\begin{mathpar}
\inferrule*
    { \inferrule* {{ \inferrule*{ }{ \bstep{\Gamma'}{0}{0, 0} } \\ \inferrule*{ \bstep{\Gamma'}{x}{0, x} }{ \bstep{\Gamma'}{\pre{x}}{\perp, (\arrow{0}{\pre{x}}) }}}}
      { \bstep{\Gamma'}{\arrow{0}{\pre{x}}}{0, (\arrow{0}{\pre{x}})} } 
    }
    { \bstepb{\Gamma}{(x = \arrow{0}{\pre{x}})}{(x = \arrow{0}{\pre{x}}), \Gamma' = \Gamma \Uparrow_x^0} }
\end{mathpar}

\begin{lemma}
\label{lemma:eqs-determinsm}
Under the assumption that $\bstep{\Gamma}{e}{v,\, e'}$ defines $v$ and $e'$ uniquely, the evaluation defined by \textsc{Eqs''} is also unique, i.e., if $\bstepb{\Gamma}{(p = e)}{(p = e_1'),\, \Gamma_1'}$ and $\bstepb{\Gamma}{(p = e)}{(p = e_2'),\, \Gamma_2'}$ then $e_1' = e_2'$ and $\Gamma_1' = \Gamma_2'$. This generalizes to the rule for multiple equations (\textsc{Eqs}).
\end{lemma}

\begin{proof}
The rule \textsc{Eqs''} can only lead to diverging results if there are multiple valid candidates $\Gamma'$ which bind the current value of the sequence being defined. Since the value of $e$ cannot depend causally on any variable bound in $p$ (i.e., any reference on the right-hand side of the equation can only refer to variables in the left-hand side under a \textbf{pre} operator), the value is uniquely defined ($v_1 = v_2$), and therefore the updated environment must be as well. The general \textsc{Eqs} rule follows by induction on the list of equations. \qed
\end{proof}

\begin{theorem}
    The step evaluation relation $\bstep{\Gamma}{e}{v,\, e'}$ is deterministic, i.e., if $\bstep{\Gamma}{e}{v_1,\, e_1'}$ and $\bstep{\Gamma}{e}{v_2,\, e_2'}$ then $v_1 = v_2$ and $e_1' = e_2'$.\label{thm:step-determinism}
\end{theorem}

\begin{proof}
    By structural induction on the syntax of expressions. Most cases are trivial, the case for the \textsc{App} rule follows from Lemma~\ref{lemma:eqs-determinsm}. 
\qed \end{proof}

\subsection{Coordination layer}

With the semantics for expressions in place, we can now define the semantics of the coordination layer. While it is possible to give the semantics of the coordination layer as a textual rewriting calculus as well, it is easier to define and understand as a graph-rewriting system.

\begin{figure}[h]
    \centering
    \begin{subfigure}[t]{0.35\textwidth}
        \centering
        \includegraphics[width=0.6\textwidth]{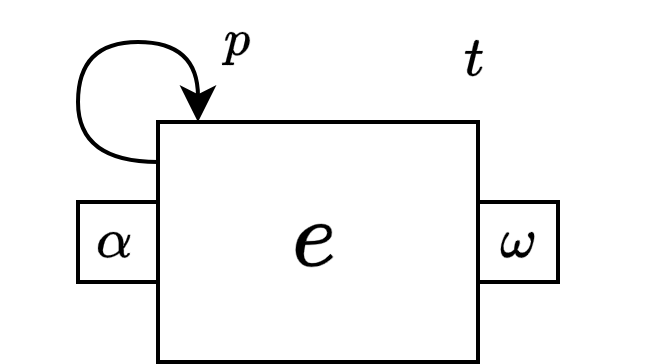}
        \caption{Node}
        \label{fig:graph-symbols-node-fifo}
    \end{subfigure}
    \hfill
    \begin{subfigure}[t]{0.35\textwidth}
        \centering
        \includegraphics[width=0.6\textwidth]{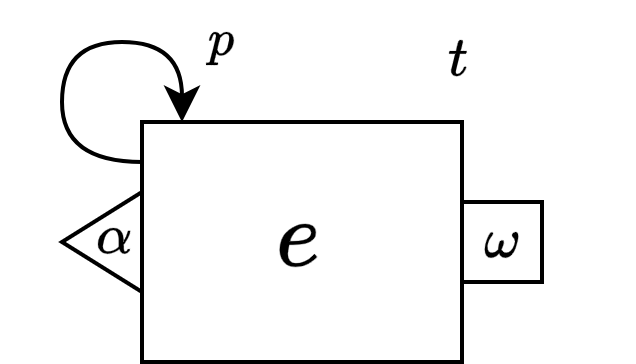}
        \caption{Node with optional input}
        \label{fig:graph-symbols-node-opt}
    \end{subfigure}
    \hfill
    \begin{subfigure}[t]{0.25\textwidth}
        \centering
        \includegraphics[width=\textwidth]{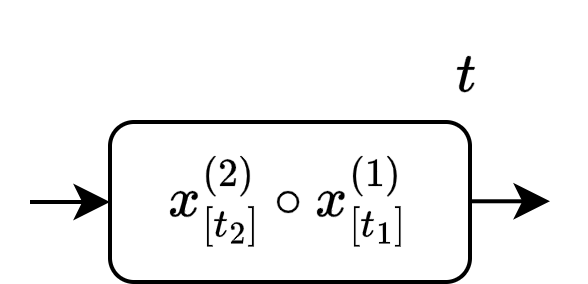}
        \caption{Channel}
        \label{fig:graph-symbols-fifo}
    \end{subfigure}
    \caption{Examples of graph symbols for the coordination layer semantics}
    \label{fig:graph-symbols}
\end{figure}

Figure \ref{fig:graph-symbols} shows examples of the used graph symbols. The first example shows a node with period $p$, which tries to execute expression $e$ at its next \emph{activation time} $t$. It communicates with other nodes through inputs taken from a channel at input $\alpha$, and puts its result in a channel at output $\omega$. Figure \ref{fig:graph-symbols-node-opt} shows the same node with an optional input $\alpha$ instead.

Figure \ref{fig:graph-symbols-fifo} depicts a channel with two elements inside. It has a time-tag $t$, which we refer to as its \emph{validity time}. Each element inside the channel also has a time tag assigned to it. The validity time of a channel is intended to be the \emph{earliest possible time-tag} of the next value written into it. We write the elements inside a channel as a sequence $\rho = x^{(n)}_{[t_n]} \circ x^{(n-1)}_{[t_{n-1}]} \circ \cdots \circ x^{(1)}_{[t_1]}$, with the following invariants:

\begin{itemize}
    \item $\forall i.\, t_i \leq t$, i.e., the time-tags of all elements in the channel must be smaller or equal to the validity time of the channel.
    \item Elements inside the channel are always ordered according to their time-tags, i.e., the oldest item (the one with the smallest time-tag) is the right-most item in the sequence.
\end{itemize}

In the current version of Mimosa, a node can only write either none or exactly one element into each output channel. In Section~\ref{sec:conclusion} we discuss how to eliminate this restriction, in which case multiple elements in a channel may have the same time-tag. This is not a problem for the semantics, as the insertion order is preserved by the order of elements in the sequence representing a channel. By slight abuse of notation, we use the infix operator $\circ$ both for appending elements on either end of the sequence%
, and for forming sequences out of individual elements themselves.

\begin{figure}[tp]
    \centering
    \includegraphics[width=0.6\linewidth]{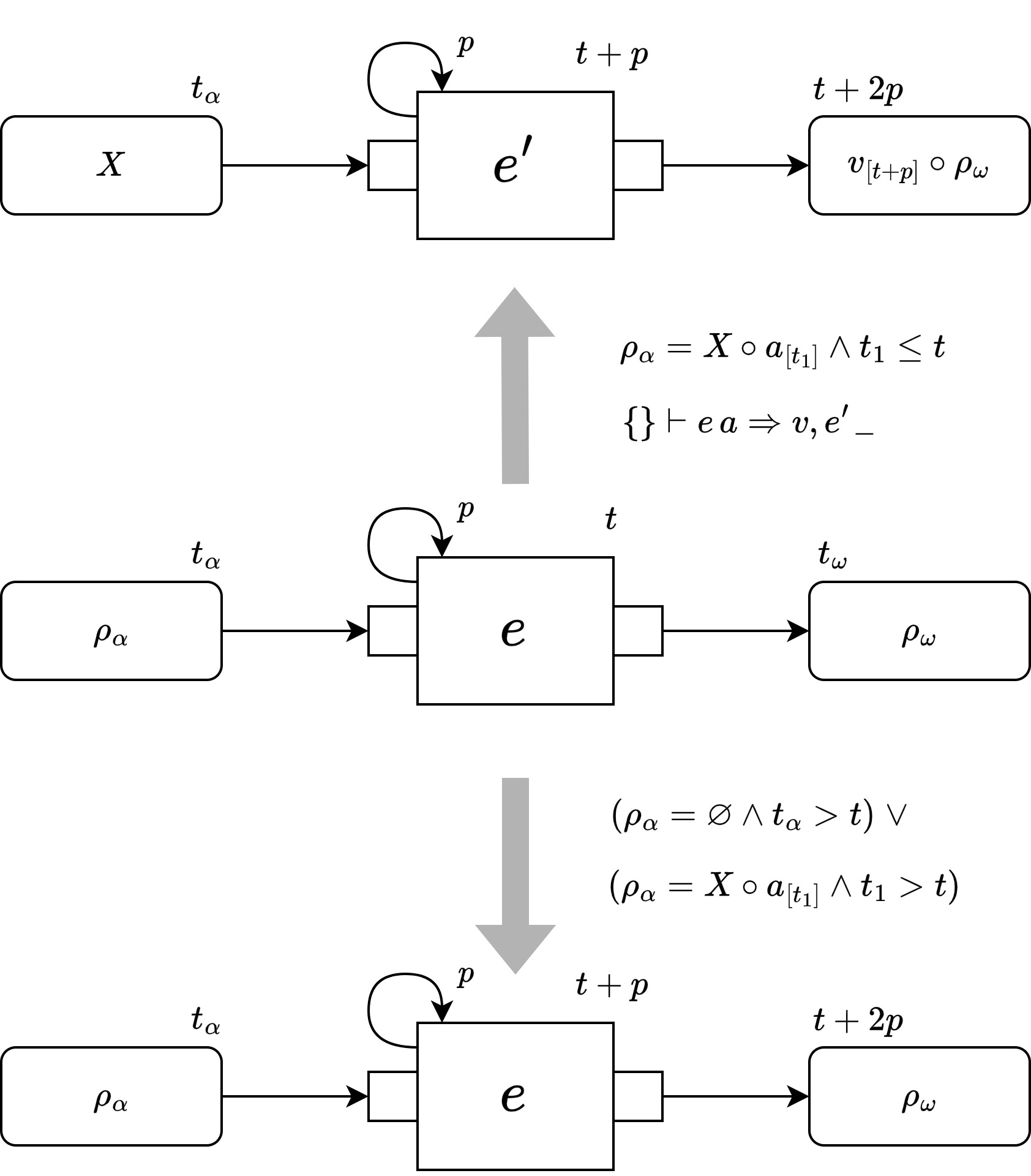}
    \caption{Rewriting rule for FIFO input.}
    \label{fig:rewrite_fifo}
\end{figure}

To ease comprehension, we only present the rewrite rules for nodes with one input and one output port. The general rules for multiple input/output are shown in Appendix~\ref{app:node-rewriting}. Figure~\ref{fig:rewrite_fifo} illustrates the rewriting rules for a node connected to two channels through its input and output port. The upper rule expresses, that if at time $t$ (i.e., the time at which the node tries to execute next) the input channel has a right-most (i.e., oldest) element $a_{[t_1]}$ for which $t_1 \leq t$ (i.e., the time-tag of the element $a$ is not in the future of the node), then the expression $e\, a$ can be evaluated. This leads to a value $v$, which is put in the output channel with time-tag $t+p$ (i.e., the end of the current period), and the validity time of the output channel can be updated to $t + 2p$, which is the earliest time at which a new element may be put into the channel. After the rewrite step, the input channel has its right-most element removed, the expression of the node is updated, the next activation time set to $t+p$, and the expression of the node rewritten to $e'$ (we can ignore the rewritten argument).

The lower rule expresses an idle step. If the input channel is empty, but its validity time $t_\alpha$ is larger than the activation time $t$ of the node, or if the right-most element has a time-tag larger than $t$ (i.e., from the perspective of the node it is in the future), then the next activation time of the node can be set to $t+p$, and the validity time of the output buffer can be updated to $t+2p$.

\begin{figure}[tp]
    \centering
    \includegraphics[width=0.6\linewidth]{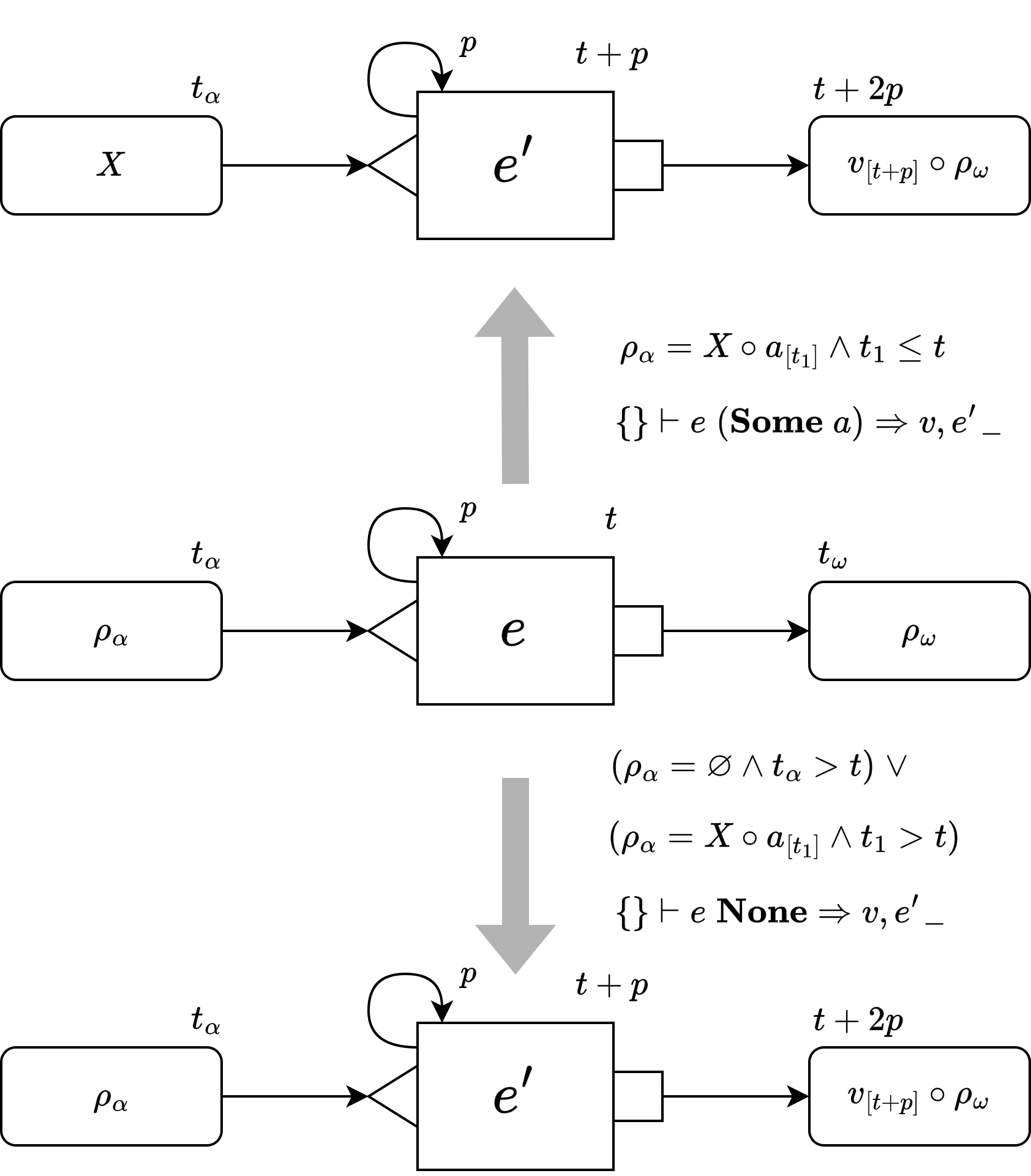}
    \caption{Rewriting rule for optional input.}
    \label{fig:rewrite_opt}
\end{figure}

Figure~\ref{fig:rewrite_opt} shows the rewriting rules for a node with optional input, which means its activation is never blocked. In the upper case, there is a valid input in the input channel, i.e., it has a time-tag smaller or equal to the current node activation time. Therefore, the expression $e\; (\textbf{Some}\, a)$ is evaluated. In the lower case, no valid input value is present, so either the input buffer is empty but the channel is valid until the current node activation time, or the time-tag of the oldest element is in the future. In this case, the expression $e\; \none$ is evaluated. The rest is analogous to the rules presented before.

To guarantee the invariant that the validity time is the earliest time point at which a value can be written to a channel (i.e., every $a_{[t']}$ added to a channel with validity time $t$ guarantees $t'\geq t$), we initialize each channel with the first possible writing time of its respective writing node. An example of a correct initialization of a network with initially empty buffers is shown in Figure~\ref{fig:init-config-correct}.

\begin{figure}[h]
    \centering
    \includegraphics[width=0.5\linewidth]{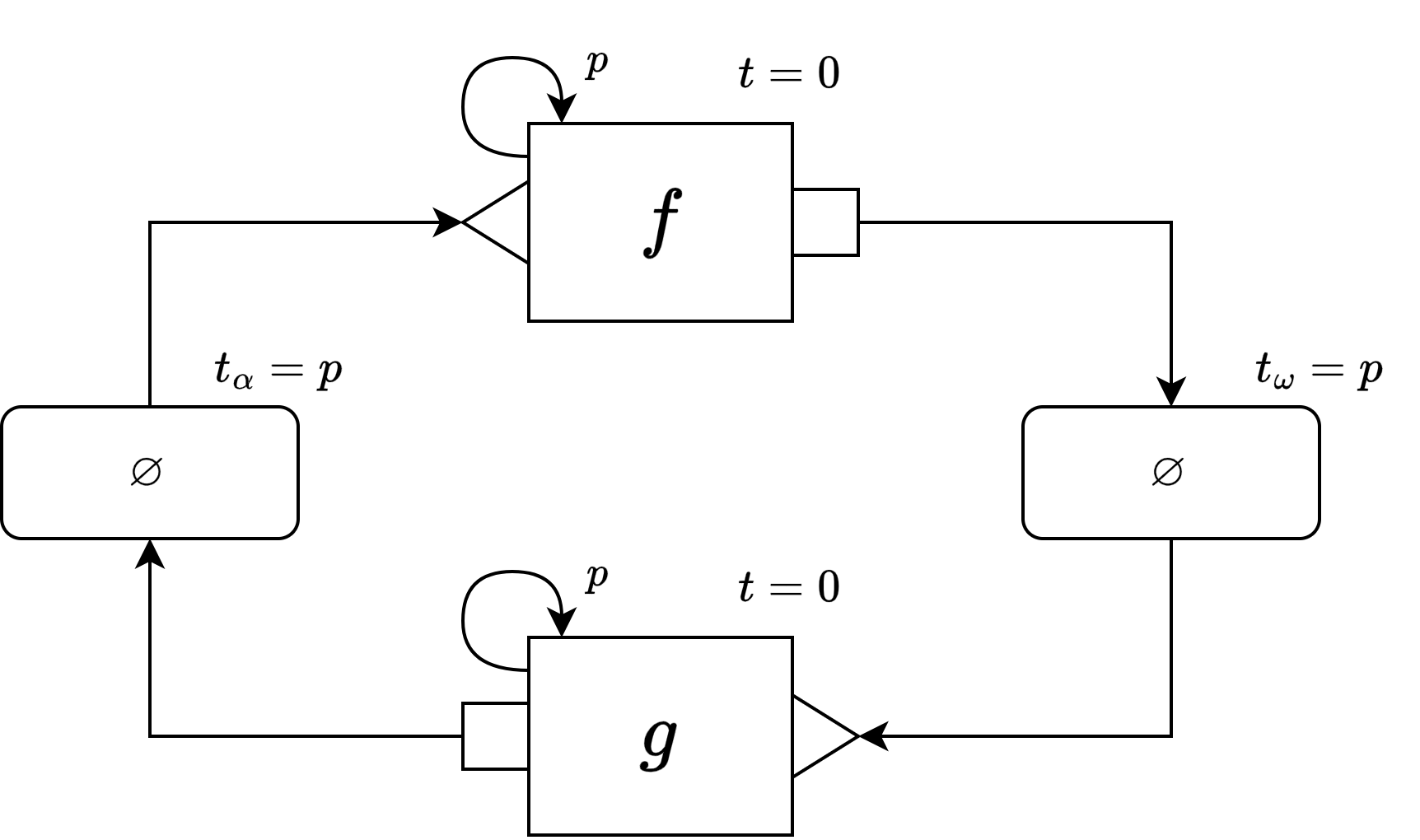}
    \caption{Example of a correct initial configuration}
    \label{fig:init-config-correct}
\end{figure}

\begin{theorem}
    The rewriting semantics described above is confluent, i.e., every possible sequence of rewritings leads to the same (timed) history of values appearing in each channel.
\end{theorem}

\begin{proof}
The rewriting rules are local (two possible rewriting candidates can only overlap on a common channel). Assume that there are two different nodes $A$ and $B$, where $A$ outputs into a channel $\sigma$ which is the input to $B$, and that $t_B$ is the next activation time for $B$. Let $t_\sigma$ be the validity time of the common channel, which is always the earliest time point at which $A$ can write, an invariant maintained by the rules. If $\sigma = \varnothing$, but $t_\sigma > t_B$, then executing $B$ first causes an idle step. The same happens if node $A$ is executed first, as the output of $A$ can only appear at or after $t_\sigma$ (it cannot be earlier than the already established validity time of the channel). The case where $\sigma$ is not empty is trivial, as executing $A$ first only puts another element in the channel, which has no impact on the execution of $B$. In all other cases (e.g. $t_\sigma < t_B$) no rule can be applied to B, and A must execute first. \qed
\end{proof} 

The rules explained above generalize naturally to the case of multiple inputs and outputs, as well as to optional output ports (they only change the writing behaviour, never the activation conditions for a node). The coordination layer semantics makes no assumption about the order of rewriting, it therefore covers all possible scheduling policies of the network of nodes\footnote{note that, strictly speaking, we do not even need a notion of ``global time'', at least as long as there are no real-time dependent side effects. This fact may be exploited for functional verification.}.
\section{Simulation}
\label{sec:simulation}

As a first means to simulate programs written in Mimosa, we have implemented the language through a deep-embedding in the programming language OCaml. We have implemented an expression evaluation function which is a direct translation of the semantic evaluation rules given in Figure~\ref{fig:big-step}. This not only adds confidence in the correctness of the implementation, it also allowed us to find (often intricate) errors in the rules early on.

The steps introduced by step prototypes may cause side effects (such as the \texttt{print_int} in Listing~\ref{code:fib}). Therefore, care must be taken during simulation regarding their (temporal) ordering. Since the coordination semantics covers all possible scheduling orders, we have opted to implement it as a discrete-event simulator which constructs the global timeline of node executions assuming perfect timings (i.e., no jitter and no clock drift). In case, the execution of multiple nodes at the same time point leads to multiple observable side effects (such as printing), the order of these effects is undefined.

We provide a prototype compiler, which can translate a Mimosa program into the mentioned OCaml embedding. This always creates an OCaml module (similar to a package in other programming languages) with a specific signature:

\begin{minted}{ocaml}
    module type Simulation = sig
        type t
        val exec_ms : t -> unit
        val exec : t -> int -> unit
        val init : unit -> t
    end
\end{minted}

\vspace{0.5cm}

A particular simulation run is a value of type \mintinline{ocaml}{Simulation.t}, which is created by a call to \mintinline{ocaml}{Simulation.init ()}. After creation, the timeline of a simulation run can be advanced through \mintinline{ocaml}{Simulation.exec_ms} and \mintinline{ocaml}{Simulation.exec} by one or multiple milliseconds, respectively. 

If the Mimosa program defines prototypes (i.e., external steps), the output of the compiler is instead a functor, i.e., a function taking a module defining all the external steps through corresponding OCaml functions, and returning a module of the above given \mintinline{ocaml}{Simulation} type. This allows using the same code for both interactive simulation and automated tests by instantiating the functor with different argument modules. Examples are shown in the documentation of the compiler~\cite{huber-mimosa-artefact}.
\section{Related Work}
\label{sec:related-work}

Lingua Franca (LF)~\cite{lohstroh-lingua-franka} is a polyglot framework for the design of hard real-time systems. As it is based upon the Reactor model~\cite{lohstroh-reactors}, computation is triggered by the occurrence of events, which are time-stamped data items. As a polyglot language, LF is meant to augment other programming languages with a coordination layer based on discrete event semantics. This allows for easy integration of libraries and other software components already written in particular mainstream languages. However, it makes formal reasoning about a particular program behaviour difficult, as one would need formal semantics for each language used in the implementation.

Giotto~\cite{henzinger-giotto} is a programming language for real-time systems, where a program describes a set of periodic tasks that communicate through ports, which always keep the last value written to them. Giotto has a particular focus on execution modes, where the system behaviour is governed by some global mode. Giotto's ports offer a similar communication primitive as MIMOS registers.
In later work~\cite{kirsch-giotto-microkernel}, a microkernel has been introduced as a compilation target, which separates the execution of code interacting with the environment~\cite{henzinger-e-machine} and the scheduling of tasks in a system. A similar approach may be interesting for the implementation of Mimosa as well.

The synchronous programming paradigm introduces into software the same synchronous abstractions as those used for digital circuit design, where a global system tick abstracts over the timing characteristics of the underlying electronics. By shifting the timing domain from physical real-time to logical ticks, temporal reasoning about the behaviour of a program becomes easier, and correctness can be verified independently of the execution platform, as long as the execution platform can guarantee the synchronous hypothesis. This becomes difficult on heterogenous or distributed hardware, where jitter and clock-drift may happen. The PALS (Physically Asynchronous Logically Synchronous)~\cite{sha-pals} protocol has been proposed as an implementation strategy, where the programmer can use the synchronous abstraction during implementation, and a middleware runtime~\cite{al-nayeem-pals-middleware} takes care of the execution on distributed hardware. MIMOS, and by extension Mimosa, avoids these problems by shifting to an asynchronous model of computation in the first place.

The Timed-C compiler~\cite{natarajan-timed-c} is a source to source compiler, which extends the C programming language with a set of timing primitives. It defines tasks which communicate through channels, similar to the ones proposed for MIMOS, and compiles to plain C on top of a real-time operating system. As a general purpose language, C is very expressive, however, it is also known for its intricacies, which makes reasoning about the behaviour of C programs difficult. We advocate, that a language for embedded systems rarely needs this degree of expressiveness, and that restricting the language to a small, well-defined kernel simplifies the verification of systems programmed in it.

\section{Conclusion and future work}
\label{sec:conclusion}

In this work, we have introduced Mimosa, a new programming language for embedded systems software on top of the MIMOS computational model. The focus of this paper is on the definition of a formal semantics, and therefore, the presented language kernel has been kept minimal. To facilitate the implementation of real systems in Mimosa, the language is intended to be extended in multiple ways:

The current implementation of Mimosa is a prototype simulator which allows for defining test cases and experimenting with the language. We are working on a compiler that translates Mimosa expressions into equivalent C code and translates the coordination layer into a set of tasks for a real-time operating system. This will allow us to run Mimosa programs on real embedded hardware. For simulation purposes, we can assume the channels to have infinite capacity, for running Mimosa programs on real hardware it would be advantageous to know the bounds of the channels beforehand. Part of the MIMOS project~\cite{wang-mimos-tool} is therefore the development of algorithms for this kind of buffer-size estimation.

This paper presents a minimal language kernel. For a fully-fledged programming language, Mimosa needs to be extended with a module system, a standard library, and potentially a more expressive type system (including enumeration and record types). The concrete syntax of Mimosa currently only allows for function abstraction at the top-level, even if the abstract syntax already includes the potential for nested functions. Extending the syntax to include (anonymous) function definitions would allow us to treat functions as first-class values, which can even be transmitted through channels.

To keep the graphical rewriting rules simpler, the current coordination layer semantics does not include registers defined in the original MIMOS paper~\cite{wang-mimos}. Registers always keep only the latest value written to them, and are therefore particularly useful as a communication link between sensors and controllers, as the sensor can then run at a higher (or lower) frequency. Registers never block the activation of a node, the only addition to the rewriting rules concerns how values are read from them. Registers can be modelled analogously to channels, where reading does not remove a value from the channel, but only looks up the value item with the largest time-tag smaller or equal than the current activation time of the node.

The optional ports presented in this paper are actually a special case of the more general \emph{up-to} ports~\cite{wang-mimos-tool}, which allow reading or writing multiple data items from or to a channel. The current coordination layer semantics can be extended to cover these types of reading and writing strategies as well, which eases the communication between nodes with different periods. The expression language would then need to be extended with operators on lists of values (such as the typical \emph{map} and \emph{fold} operators known from functional programming languages). This is similar to the extension of Lustre with array iterators~\cite{morel-lustre-array-iterators}.

In Mimosa, certain expressions, such as conditionals, only selectively evaluate their sub-expressions, which in turn removes the referential transparency of variable names. We have chosen this design to provide a similar mechanism to Lustre's multi-clocks, without the added mental complexity. It also lays the foundations for facilitating expression evaluation with side effects. The only effect we currently model is memory, however, recent programming languages~\cite{madsen-flix-lang,brachthaeuser-effekt-lang,reinking-perceus} have shown the advantage of being able to express side effects of a program in terms of types. Being able to track the effects an expression may exhibit during evaluation offers new opportunities for optimization (e.g., application of functions to constants during compile time). It also allows for speculative execution, where a node without observable side effects can execute even before its release time in case it has all its required inputs. This can ultimately lead to better utilization of processors, while keeping the functional output of the system the same. We therefore intend to extend a future version of Mimosa with effect types.

\begin{credits}
  \subsubsection{\ackname} This work was partially funded by ERC through project CUSTOMER and by the Knut and Alice Wallenberg Foundation through project UPDATE.
\end{credits}

\bibliographystyle{splncs04}
\bibliography{references}

\newpage

\appendix
\section{Initialization analysis}
\label{app:init-analysis}

\emph{Initialization analysis} is part of the frontend of the Mimosa compiler~\cite{huber-mimosa-artefact}. It is responsible for proving that the undefined values at the start of sequences, which are introduced by \textbf{pre} operators, do not affect the output of a step.

The analysis performed by the Mimosa compiler is similar to the one in Lustre~\cite{colaco-lustre-init}, however, the selective evaluation of certain sub-expressions leads to different requirements for proper initialization.

For example, in Lustre the following holds:
\begin{align*}
    \ite{a}{(\arrow{0}{b})}{(\arrow{0}{c})} \equiv \arrow{0}{(\ite{a}{b}{c})}
\end{align*}

This does not hold in Mimosa, as the two branches are selectively evaluated depending on the value of the condition expression. In general, both branches of a conditional expression need to be properly initialized.

In Lustre, \textbf{pre} statements may also be nested to refer to values of a stream from multiple cycles before:
\begin{align*}
\arrow{0}{\arrow{0}{\pre{\pre{x}}}}
\end{align*}

This expression leads to an undefined value at the second cycle in Mimosa:
\begin{mathpar}
\inferrule*
{
    \inferrule*{ }{\bstep{\Gamma}{0}{0,\, \_}} \\
    \inferrule* { \inferrule*{ }{\bstep{\Gamma}{0}{0,\, \_}} \\ \inferrule* { 
        \inferrule*
        { \ldots }
        {\bstep{\Gamma}{\pre{x}}{\perp,\, \ldots}}
    }{\bstep{\Gamma}{(\pre{\pre{x}})}{\_,\, (\arrow{\perp}{\ldots})}}}
    {\bstep{\Gamma}{(\arrow{0}{\pre{\pre{x}}})}{\_,\, (\arrow{\perp}{\ldots})}}
}
{\bstep{\Gamma}{(\arrow{0}{\arrow{0}{\pre{\pre{x}}}})}{0,\, (\arrow{\perp}{\ldots}})}
\end{mathpar}

Nested \textbf{pre} statements are nevertheless allowed in Mimosa, $\bm{\rightarrow}$ and $\pre$ need to be used alternately, which leads to the expected behaviour:
\begin{mathpar}
\inferrule*
{
\inferrule*{ }{\bstep{\Gamma}{0}{0,\, \_}} \\
\inferrule*{ 
    \inferrule*{ 
        \inferrule*{ }{\bstep{\Gamma}{0}{0,\, \_}} \\
        \inferrule*{ \inferrule*{ }{\bstep{\Gamma}{x}{\Gamma \Downarrow_x,\, x}} }{ \bstep{\Gamma}{\pre{x}}{\perp,\, (\arrow{\Gamma \Downarrow_x}{\pre{x}})} }
    }{ \bstep{\Gamma}{(\arrow{0}{\pre{x}})}{0,\, (\arrow{\Gamma \Downarrow_x}{\pre{x}}}) }
}{ \bstep{\Gamma}{\pre{(\arrow{0}{\pre{x}})}}{\perp,\, (\arrow{0}{\arrow{\Gamma \Downarrow_x}{\pre{x}}})} }
}
{
\bstep{\Gamma}{(\arrow{0}{\pre{(\arrow{0}{\pre{x}}})})}{0,\, (\arrow{0}{\arrow{\Gamma \Downarrow_x}{\pre{x}}})}
}
\end{mathpar}

Given this requirement, each sequence in Mimosa can only have one of two possible initialization types: initialized or uninitialized. In order to further simplify the reasoning about programs, we also require that step inputs and outputs are always initialized.
\newpage
\section{Node-level rewriting rules}
\label{app:node-rewriting}

\begin{figure}[h]
    \centering
    \includegraphics[width=0.9\linewidth]{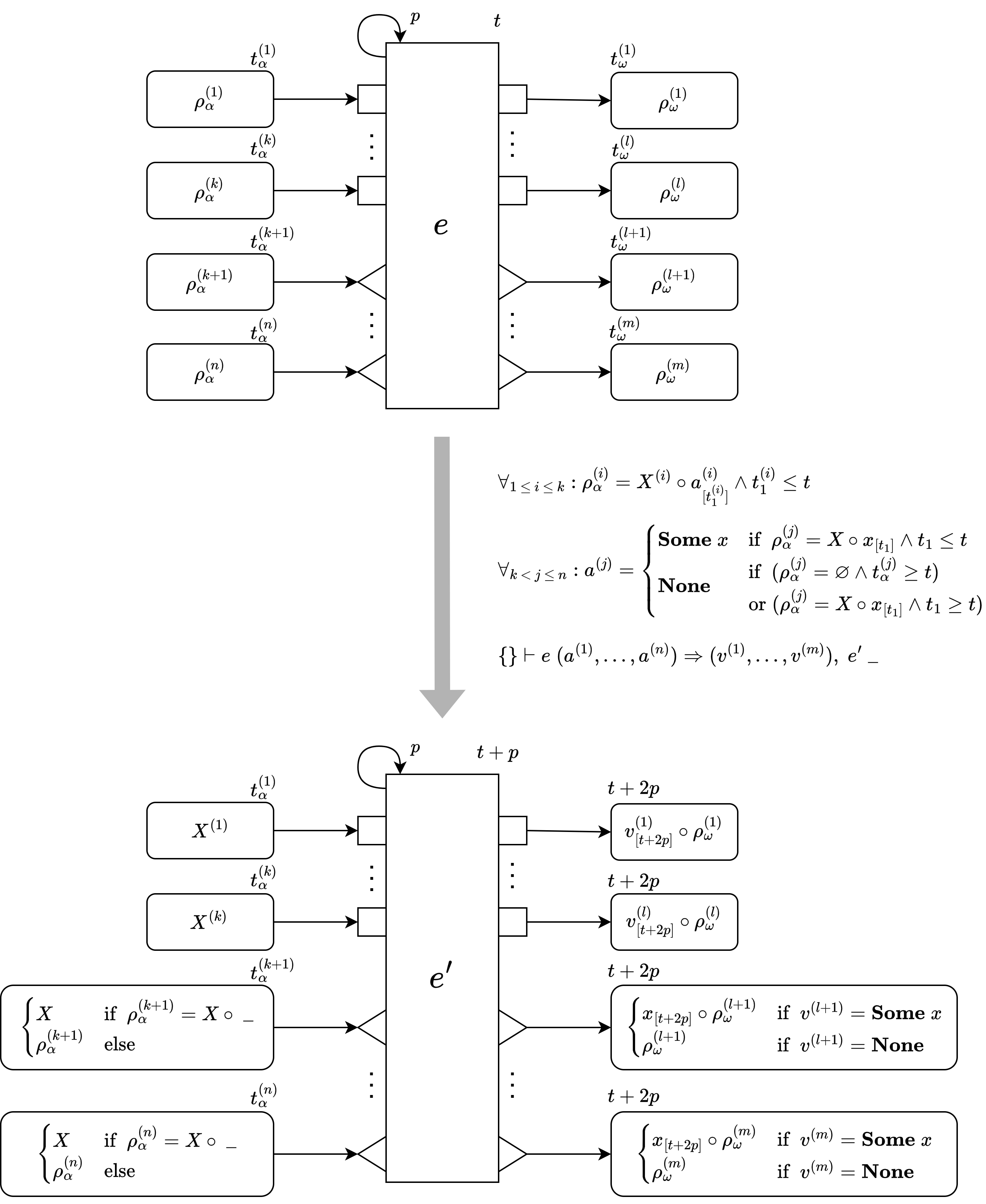}
    \caption{General rewriting rule for node execution.}
    \label{fig:rewriting-exec}
\end{figure}

\begin{figure}[h]
    \centering
    \includegraphics[width=0.9\linewidth]{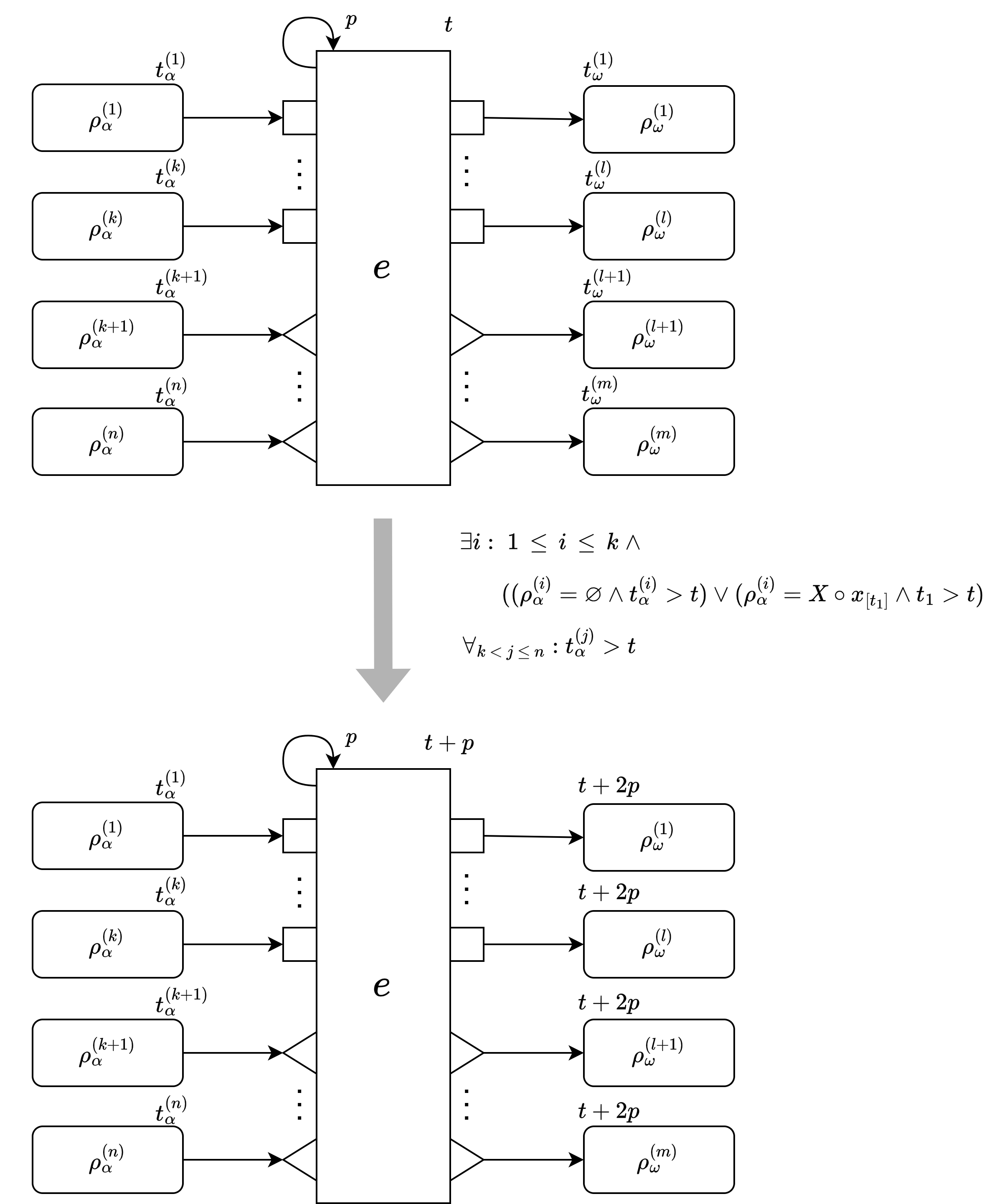}
    \caption{General rewriting rule for idle step.}
    \label{fig:rewriting-idle}
\end{figure}
\end{document}